\documentclass[12pt,leqeq]{article}
\usepackage{amssymb,amsthm}
\usepackage{amsmath}
\usepackage{amscd}
\usepackage{xy}
\xyoption{all}
\usepackage[dvips]{graphicx}
\newtheorem{thm}{Theorem}
\newtheorem{lem}[thm]{Lemma}
\newtheorem{cor}[thm]{Corollary}
\newtheorem{prop}[thm]{Proposition} 
\newtheorem{rem}[thm]{Remark}
\newtheorem{conj}[thm]{Conjecture}

\newtheorem{exmp}[thm]{Example}

\date{}
\begin{document}
\setlength{\baselineskip}{16pt}
\title{On the stability of the PWP method}
\author{Rafael D\'\i az \ \ and \ \ Angelica Vargas}

\maketitle

\begin{abstract}
The PWP method was introduced by D\'iaz in 2009 as a technique for measuring indirect influences in complex networks. It depends on a matrix $D$, provided by the user, called the matrix of direct influences,  and  on a positive real parameter $\lambda$ which is part of the method itself. We study changes in the method's predictions as $D$ and $\lambda$ vary.
\end{abstract}

\section{Introduction}

One of the main problems in network theory is to define a ranking on the vertices of a network reflecting the importance that each vertex plays in the network. Thus, for the case of  weighted directed  networks, one is looking for
maps $$ \mathrm{wdigraphs} \ \ \longrightarrow \ \ \mbox{rankings on vertices}$$
sending a weighted directed graph to a ranking on its set of vertices. Asking for a linear order on vertices is clearly too much since there could be vertices playing equally important roles in the network. A ranking on a set $X$ is a pre-order defined by a map $r:X \longrightarrow \mathbb{R}$ such that
$x\leq y\ $ if and only if $\ r(x) \leq r(y).$ Clearly, different maps $r$ may give rise to the same pre-order, i.e. to the same ranking.\\

Finding a suitable map as above is not an easy task, and it is probably a problem with no universal solution. A first approach to this problem is via the total degree map $$\mathrm{G}: \mathrm{wdigraphs} \ \ \longrightarrow \ \ \mbox{rankings on vertices} $$
which orders vertices according to their total degree, i.e. the sum of weights of edges reaching or leaving a given vertex. With this ordering strongly connected vertices are deemed as having greater importance for the network.\\

We work with networks of influences. As it often happens in mathematics,  it is better to leave some terms undefined and let examples tell us the intended meaning. So, allow us to review a few of examples of networks of influences.\\

   \noindent \textbf{Author Citation Networks.}\\

   Vertices in these networks are authors of scientific publications. An author $j$ has exerted an influence on author $i$, if there is at least one publication of $i$ citing a publication of $j$. The matrix of direct influences is given by
   $$D_{ij} \ = \ \frac{\sharp \ \mbox{publications of $i$ citing a publication of $j$}}
   {\sharp \ \mbox{publications of  $i$}}.$$ The PWP method is designed to study indirect influences in networks of this sort.\\

   \noindent \textbf{Business Providers  Networks.}\\

   Vertices are businesses in a group or economic sector. A business $j$ influences business $i\ $ if $\ j$ provides products (or services) to $i.$ The matrix of direct influences is given by
   $$D_{ij} \ = \ \ \frac{\mbox{amount spend by bussines $i$ buying products from bussines $j$} }{\mbox{budget of bussines $i$}} .$$

   \noindent \textbf{International Trade Network.}\\

   Vertices are countries. A country $j$ exerts an influence on the economy of country $i$ if there is trade between $i$ and $j$. Thus the network itself is undirected, but the weight on edges do take direction into account.  The matrix of direct influences is given by
   $$D_{ij} \ = \ \frac{I_{i,j} \ + \ E_{i,j}}{C_i} ,$$
   where $\ I_{i,j}\ $ is the total amount that $i$ imports from $j, \ $ $E_{i,j}\ $ is the total amount that $i$ exports to $j$, and $C_i$ counts the total amount of international trade of country $i$ (total of imports plus total of exports.) This example have been studied by D\'iaz and G\'omez \cite{diazgomez}.\\

   \noindent \textbf{Process-Matter  Networks.}\\

   In these networks we have two types of vertices: black vertices for processes and white vertices for matter. The network represents a production system consisting of several processes. Each process takes some materials as input, and produces other materials as output, which may in turn be the inputs of other processes, and so on ... A matter vertex $m$ influences a process vertex $p$ if $m$ is one of the inputs that $p$ needs to operate; a process vertex $p$ influences a matter vertex $m$, if $m$ is one of the outputs that $p$ produces.\\

   Clearly these type of networks can be use to model a host of phenomena: chemical networks, metabolic networks, business organization systems, etc.  The choice of weights for a process-matter network depend on the intended application, but they will always have the form:
   $$\left(
       \begin{array}{cc}
         0 & M \\
         P & 0 \\
       \end{array}
     \right)
    $$
   where indices for the first block of columns represent processes, indices in the second block of columns represent materials, the matrix
   $P$ represents the influences of processes on materials, and the matrix $N$ represents the influences of materials on processes. \\

   It is often useful to study influences among processes themselves, and also influences among materials themselves. This can be readily achieved by using, respectively, the product matrices
   $$ MP \ \ \ \ \ \mbox{and} \ \ \ \ \ PM.$$
   In words, a process $p$ influences a process $q$ if an  output of $p$ is used by $q$ as input, and similarly a material $m$ influences a material $n$ is there is at least one process that uses $m$ as input and produces $n$ as output. Figure \ref{pm} shows an example of a process-matter network. Figure \ref{hoy} displays the associated networks of influences among processes themselves, and among materials themselves.  \\

\begin{figure}
\centering
\includegraphics[width=0.4\textwidth]{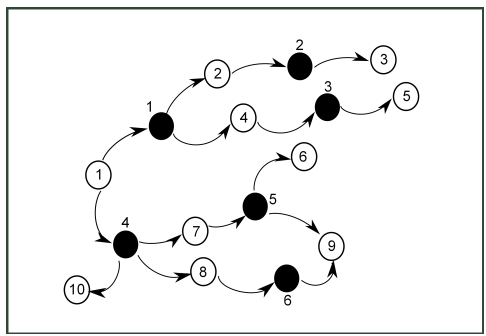}
\caption{A Process-Matter Network.}
\label{pm}
\end{figure}

\begin{figure}[htb]
\centering
\begin{tabular}{@{}cc@{}}
\includegraphics[width=0.4\textwidth]{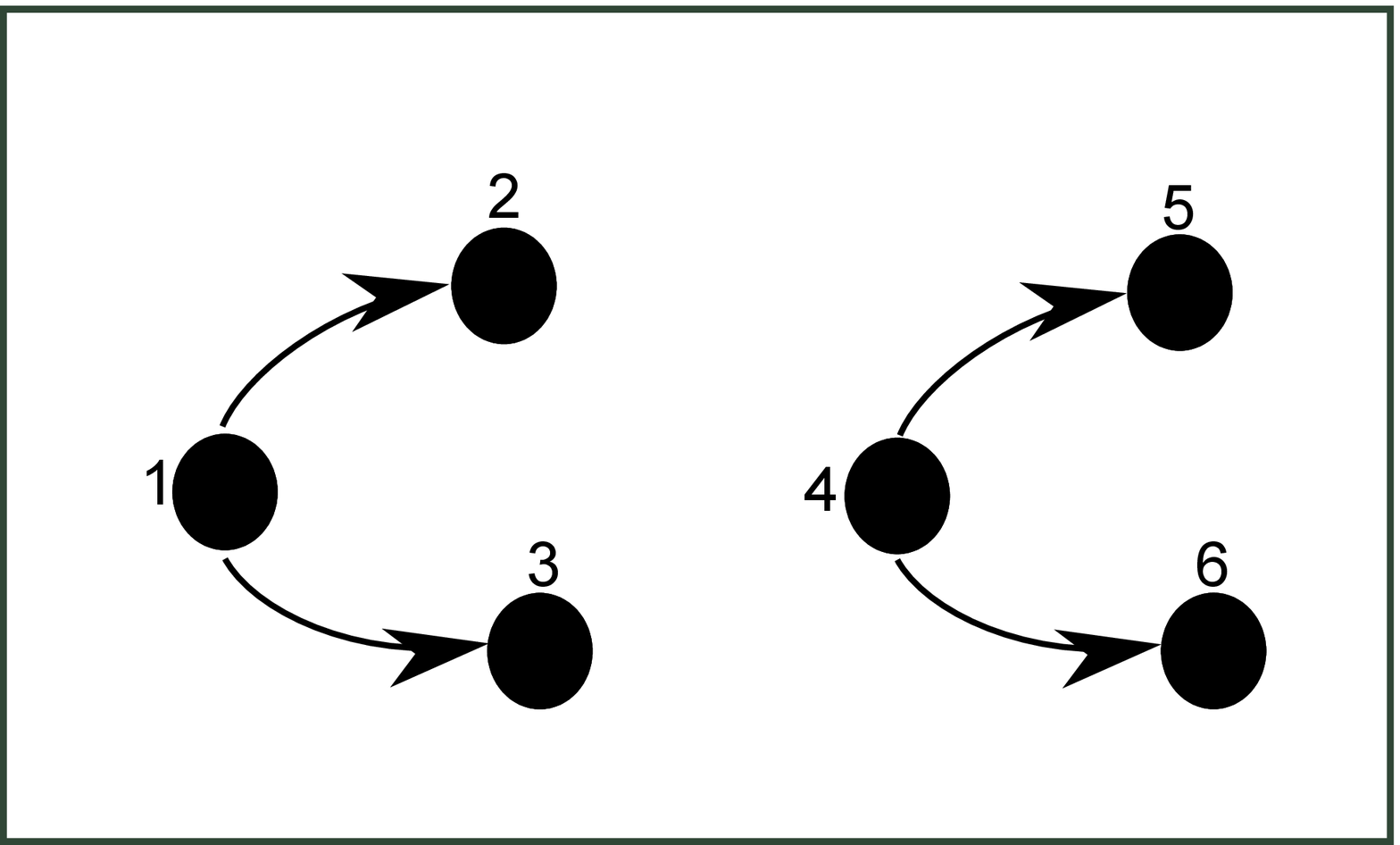}
\
\includegraphics[width=0.4\textwidth]{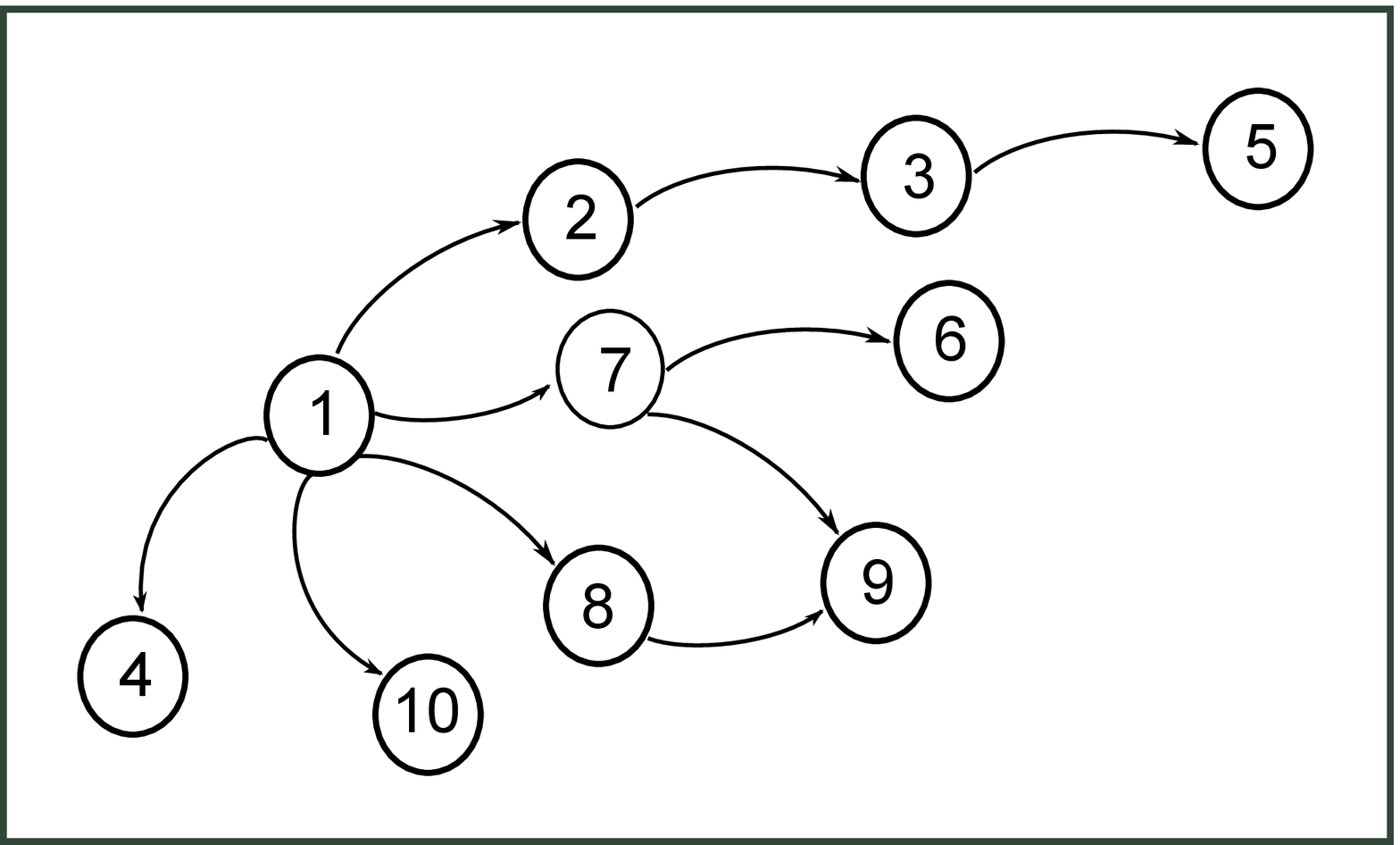}
\end{tabular}
\caption{Left: Associated Processes Network. Right: Associated Matter Network.}
\label{hoy}
\end{figure}

With these examples in mind we go back to our main problem. From the viewpoint of networks of influences, we say that the total degree map $\mathrm{G}$ defines a ranking of vertices based on their direct importance. Although of interest, ranking by total direct degree fails to acknowledge that some direct influences may be short lived (if exerted over an isolated vertex) or may be enhanced (if exerted over a highly connected vertex.) Thus although $\mathrm{G}$ should not be overlooked, it should not be regarded as the unique or final answer.\\

A deeper  approach to our problem should take  indirect influences into account; however there are many alternative ways for doing so. To proceed forward we look for a map $$R: \mathrm{wdigraphs} \ \ \longrightarrow \ \ \mbox{rankings on vertices}$$ of the form:
$$ \mathrm{wdigraphs} \ \ \overset{\mathrm{T}}{\longrightarrow} \ \ \mathrm{wdigraphs} \ \ \overset{\mathrm{R}}\longrightarrow \ \ \mbox{rankings on vertices}$$ where the map $\mathrm{T}$ transforms a network of direct influences into a network of indirect influences, and $\mathrm{L}$ is any ranking method computed from direct influences. \\

To be able to pick a particular map $\mathrm{T}$ among the many possibilities some choices must be made.  The PWP method for counting indirect influences \cite{diaz} was founded over the following principles:

\begin{enumerate}
  \item Indirect influences arise from the concatenation of direct influences.
  \item Indirect influences do not arise in any other way.
  \item The weight of a concatenation of direct influences is proportional to the product of the weight of the direct influences that it comprehends, suitable modified to be compatible with our next principles.
  \item As a rule, the longer a concatenation of direct influences, the lesser the indirect influence exerted by it.
  \item Stochastic direct influences should generate stochastic indirect influences.
  \item Robust convergency properties are expected.
  \item The map should be equivariant under to conjugation.
\end{enumerate}

The PWP method is defined via the mapping
$$ \mathrm{wdigraphs} \ \ \overset{\mathrm{T}}{\longrightarrow} \ \ \mathrm{wdigraphs}, $$
which can be described in simpler terms as a map
$$ \mathrm{M}_n(\mathbb{R}) \ \ \overset{\mathrm{T}}{\longrightarrow} \ \ \mathrm{M}_n(\mathbb{R}),  $$
from real square matrices of size $n$ to itself, since we can identify (simple) weighted directed graphs with their adjacency matrices once a linear order has been fixed on their vertices. \\

The map $\mathrm{T}$ depends on a real parameter $\lambda >0,$ and  is given on $D \in \mathrm{M}_n(\mathbb{R})$ by
$$\mathrm{T}(D, \lambda)  \ \ = \ \ \frac{e_+^{\lambda D}}{e_+^{\lambda}}\ \ = \ \ \frac{e^{\lambda D}  -  I}{e^{\lambda}  -  1}
\ \ = \ \ \frac{\sum_{k=1}^{\infty}D^k\frac{\lambda^k}{k!}}{\sum_{k=1}^{\infty}\frac{\lambda^k}{k!}}.$$
The map $\mathrm{T}$ satisfies our six requirements. Properties 1, 2, and 3 hold since the entries of matrix $D^k$ count weighted-paths of length $k$, i.e. concatenations of $k$ direct influences. The dividing factor $\frac{1}{k!}$ implies property 4. Property 5 justifies the inclusion of the normalizing factor $e_+^{\lambda}.$ Property 6 is insured by the factors $\frac{1}{k!}$. Property 7 can be easily verified.\\

There is a number of other good choices for the map $\mathrm{T}$, for example: the Katz index \cite{k}, the MICMAC of Godet \cite{godet}, the PageRank of Google \cite{b2, b1, meyer}, the Heat Kernel of Chung \cite{chung}, and the communicability method of Estrada and Hatano \cite{estrada}. These well-tested methods fail to satisfy some of the above requirements for the following reasons:

\begin{itemize}
\item The Katz index (1953) when suitable extended for weighted networks and normalized, satisfies most of our requirements; however its convergency properties are not as strong as desired, this is the main reason why the factors $\ \frac{1}{k!}\ $ are included in the PWP method.
  \item MICMAC (70's) considers paths of a fixed length $k,$ thus it fails to satisfy 1.
  \item PageRank (1999) computes indirect influences by first transforming the matrix of direct influences into a Markovian matrix $\tilde{D}$; in this process influences are created using a mechanism different to concatenation. Thus PageRank fails to satisfy 2.
  \item Heat Kernel (2007) comes pretty close to satisfying our principles. However, it introduces indirect self-influences not coming from the matrix $D$, thus it fails to satisfy property 2. This is the main reason why the PWP method uses the function $e_+^x$, instead of the exponential map $e^x$.
   \item Communicability (2007) corresponds (after normalization) to the $\lambda=1$ case in the heat kernel method; so the same arguments as above apply. Note that in both cases, starting with vanishing direct influences one obtains non-vanishing indirect influences. In contrast, in the absence of direct influences, the PWP method yields vanishing indirect influences, as also does the Katz index.

\end{itemize}

Pre-composing any ranking based on the matrix of direct influences with the PWP map one obtains a new ranking that takes indirect influences into account. We consider three ranking methods:

\begin{itemize}
  \item The ranking by total degree $\mathrm{G},$ which after composition with the PWP map we call ranking by importance.
  \item The ranking by outgoing degree $\mathrm{F},$ which after composition with the PWP map we call ranking by indirect influence.
  \item The ranking by incoming degree $\mathrm{E},$ which after composition with the PWP map we call ranking by indirect dependence.
\end{itemize}

  The orderings by incoming, outgoing, and total degrees may be thought, respectively, as maps
$$\mathrm{E}, \ \mathrm{F}, \ \mathrm{G}: \mathrm{M}_n(\mathbb{R})  \ \longrightarrow \  \mbox{rankings on [n]}$$ where for a matrix $D$ the place of a vertex in the respective orders is proportional to the values of the maps
$$\mathrm{E}_i \ = \ \sum_{j=1}^n D_{ij}, \ \ \ \ \ \mathrm{F}_i \ = \ \sum_{j=1}^n D_{ji}, \ \ \ \ \ \ \mbox{and} \ \ \ \ \ \ \mathrm{G}_i \ = \ \sum_{j=i}^n (D_{ij} \ + \  D_{ji}).$$
Pre-composing these maps with the PWP map we obtain, respectively, maps
$$E, \ F, \ I: \mathrm{M}_n(\mathbb{R})  \ \longrightarrow \  \mbox{rankings on [n]}$$  where for a matrix of direct influences $D$ the place of a vertex in the respective orders is proportional to the values of the maps
$$E_i \ = \ \sum_{j=1}^n T_{ij}, \ \ \ \ \ F_i \ = \ \sum_{j=1}^n T_{ji}, \ \ \ \ \ \ \mbox{and} \ \ \ \ \ \ I_i \ = \ \sum_{j=i}^n (T_{ij} \ + \  T_{ji}),$$
defined in terms of the PWP matrix of indirect influences $T=\mathrm{T}D.$ \\

Since there is only a finite number of rankings on a finite set, the maps $E, \ F$ and $I$ can be continuous only if they are constant. Our main interest in this work is to study the regions of continuity of these maps, i.e. the regions where they are constant maps. Note that the rankings imposed by $E, \ F$ and $I$ on the set of vertices depend ultimately on $D$ and $\lambda$. We consider continuity with respect to $D$ and $\lambda$ separately, as they play quite different roles: discontinuity with respect to $D$ is taken as a manifestation of data sensitivity, in itself a positive phenomena, whereas discontinuity with respect to $\lambda$ is taken as a sign of the care required in the choice of $\lambda$ in the applications. The main question is whether or not changing $\lambda$ will let the PWP method to impose a more or less arbitrary ranking among the vertices. Our results, although partial, indicate that the opposite is the case: even if several rankings may result as $\lambda$ varies, they are seldom arbitrary, and evolve following a rigid pattern.

\section{Data Sensitivity of the PWP Method}

In this section we analyze the continuity of the map $$I: \mathrm{M}_n(\mathbb{R}) \  \longrightarrow \  \mbox{rankings on [n]}$$ obtained from the ordering by total degree after an application of the PWP map $T(D, \lambda)$. We set $\lambda=1$ and consider discontinuities of the map $I$ with respect to matrix $D$ of direct influences. Note that the matrix $D$ is provided by the user, and we are going to test the sensitivity of the method to small changes in the matrix $D$. As expected, small changes in $D$ can lead to quite different rankings on $[n]$, stressing the need for a judicious choice of data. \\

 Let $\mathrm{L}_6$ be the linear graph with $6$ nodes, see Figure \ref{l}, and let $\mathrm{L}_6(\epsilon)$ be a perturbation of it obtained by adding a new edge of weight $\epsilon$ as shown in Figure \ref{lado_e}. For $\epsilon=0$ we recover the graph $\mathrm{L}_6$, further studied in Section \ref{slg}, and as $\epsilon$ grows the new edge becomes more relevant and modifies the ranking by importance of the vertices.\\

\begin{figure}[htb]
\centering
  \begin{tabular}{@{}cc@{}}
   \includegraphics[width=0.5\textwidth]{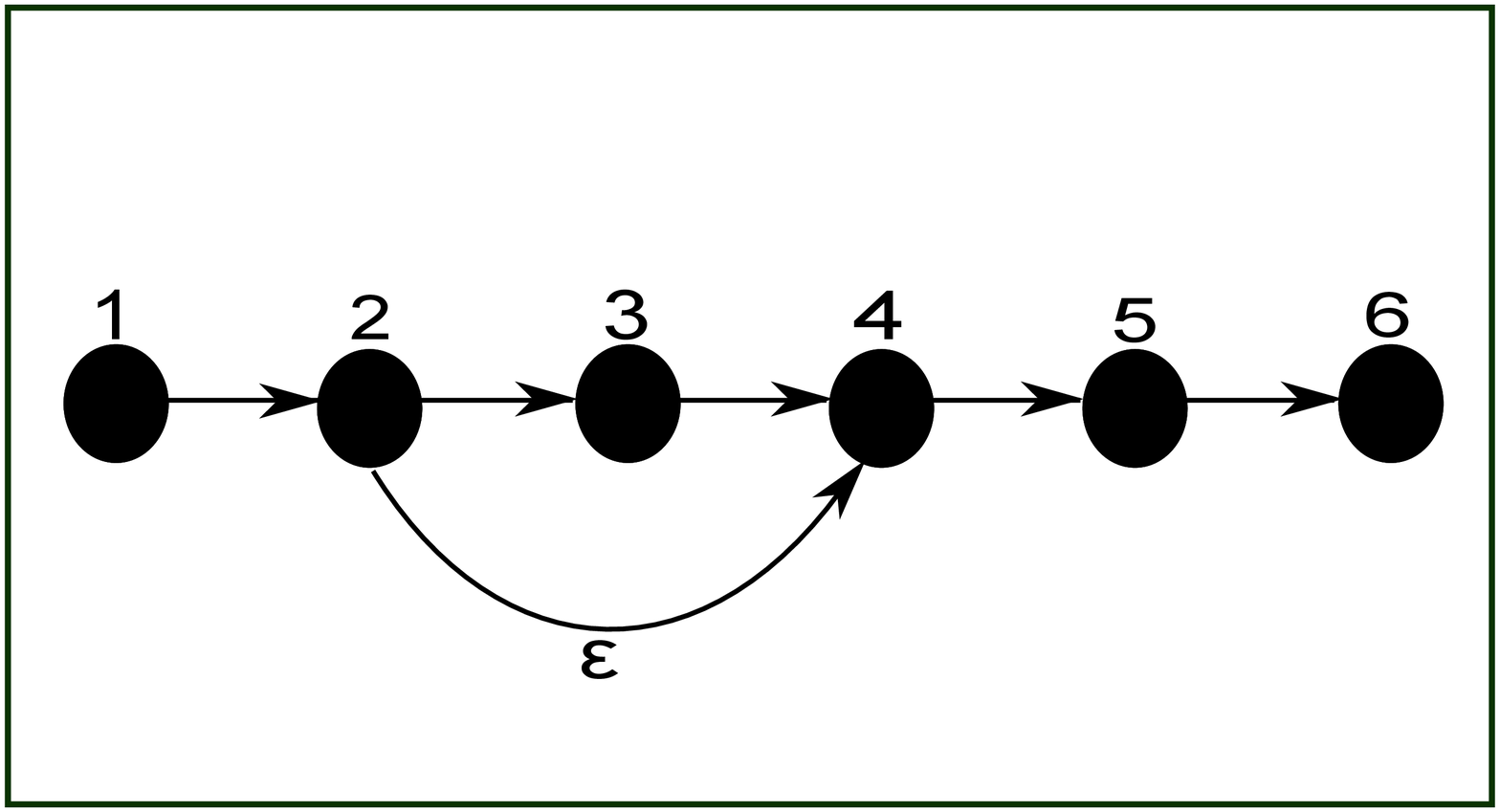}
   \
   \includegraphics[width=0.5\textwidth]{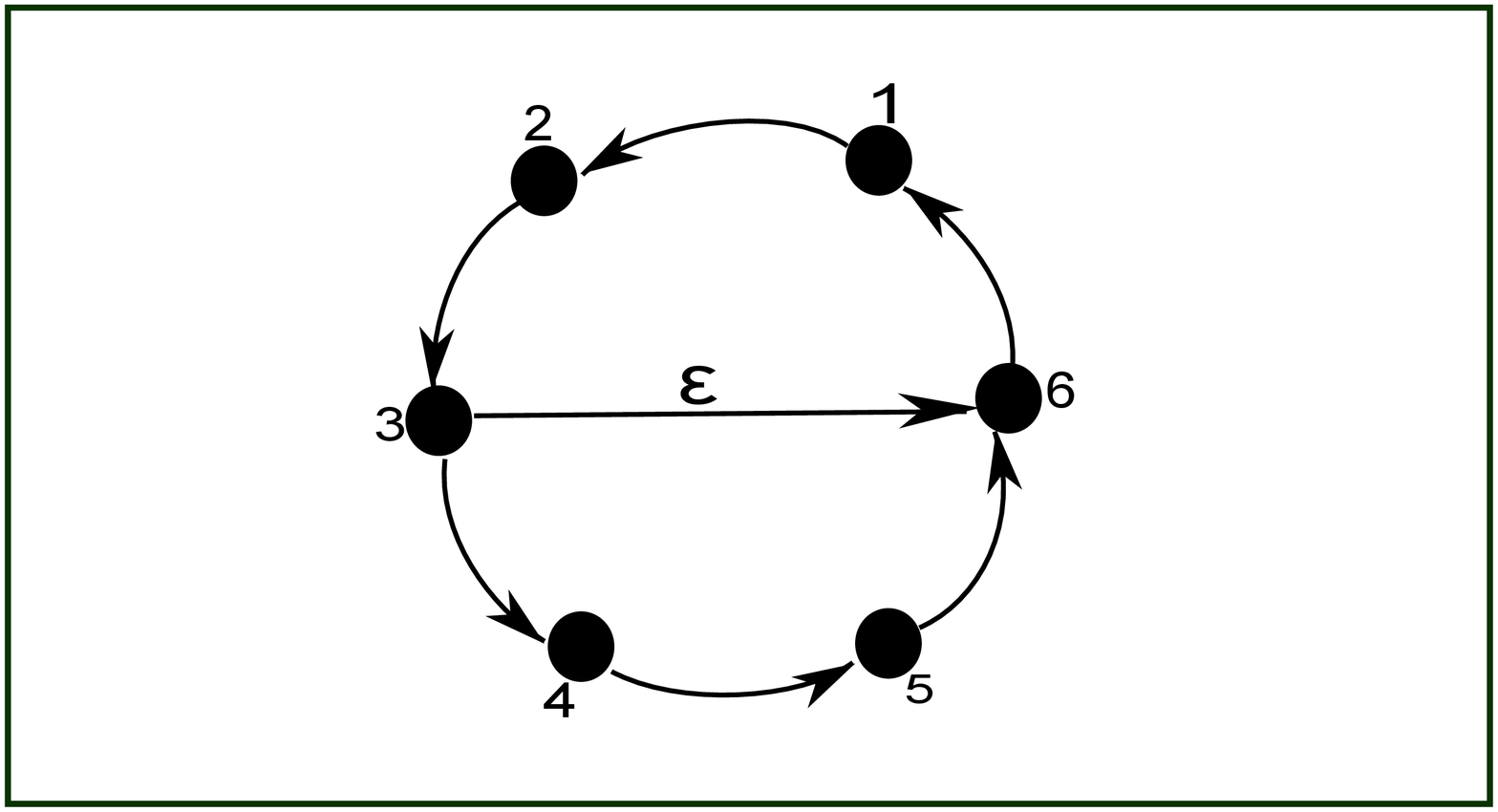}&
   \end{tabular}
   \caption{Left: Graph $\mathrm{L}_6(\epsilon)$. Right: Graph $C_6(\epsilon)$. }
   \label{lado_e}
\end{figure}

Constructed numerically, Table \ref{tabla1} shows how the ranking of vertices by importance changes as $\epsilon$ varies.
Already for $\epsilon=0.01$ we see that a change occurs, for example, the pairs of vertices $\{3,4\}, \{2,5\}$ and $\{1,6\}$ have the same importance for $\epsilon=0$, but for $\epsilon=0.01$ all of them have different importance. Further changes in ranking were located at
$\epsilon= 0.28, \ 0.69, \ 2.1, \ 2.8, \ 7, \ \mbox{and} \ 23.9.$ Note that the vertex $4$ is on the leading position up to $\epsilon=2.8,$ where the vertex $2$ becomes the most important one. Note also that vertex $3$ begins at the leading position, and for $\epsilon \geq 7$ becomes the less important vertex. After we reach $23.9$ our numerical experiments showed no further changes in ranking.\\

This fairly simple example, already shows the high data sensitivity of the ranking of vertices by importance based on the PWP method.

\begin{center}
\begin{tabular}{|c|c|}
  \hline
   $\epsilon$ & Order by Importance  \\  \hline
  0 & $3,4> 2,5 > 1,6$ \\  \hline
  0.01& $4>3>2>5>1>6$ \\  \hline
    0.28& $4>2>3>5>1>6$\\  \hline
   0.69& $4>2>5>3>1>6$\\  \hline
   2.1&$4>2>5>1>3>6$\\  \hline
    2.8&$2>4>5>1>3>6$\\  \hline
   7&$2>4>5>1>6>3$\\ \hline
    23.9&$2>4>1>5>6>3$\\  \hline
\end{tabular}\label{tabla1}
\end{center}

Let us consider a second example, the circuit graph $C_6(\epsilon)$ with $6$ nodes extended by a new edge of weight $\epsilon$ as shown in Figure \ref{lado_e}. For $\epsilon=0$ all vertices, as will be shown in Section \ref{sw}, are equally important. But as soon as $\epsilon$ reaches $0.0001$ we find the ranking $$3,6\ > \ 1,2\ >\ 4,5,$$ which, according to our numerical calculations, remains stable for higher values of $\epsilon.$ This example displays both behaviours a highly sensitive one at the beginning, followed by a fairly stable one for higher values of $\epsilon$. \\

We close this section with a remark on the impact that a change of scale, at data-level, have on the applications of the PWP method, i.e. we let our matrix of direct influences $D$ be replaced by a new matrix of direct influences $cD,\ $ with $\ c \in \mathbb{R}_+. $ We have that:
$$T(cD,\lambda) \ = \ \frac{1}{e^{\lambda}_+}\sum_{k=1}^{\infty}(cD)^k\frac{\lambda^k}{k!}\ = \ \frac{1}{e^{\lambda}_+}\sum_{k=1}^{\infty}D^k\frac{(c \lambda)^k}{k!}\ = \  \frac{e^{c \lambda}_+}{e^{\lambda}_+}T(D,c\lambda).$$
As $\ \frac{e^{c \lambda}_+}{e^{\lambda}_+}\ $ is a positive real number, it does not affect the rankings by dependence, influence or importance. Therefore in the applications of the PWP method rescaling data by a factor amounts to rescaling $\lambda$ by the same factor, and thus stability with respect to data rescaling is a particular case of stability with respect to changes in  $\lambda.$

\section{Stability on Linear Graphs}\label{slg}

In this section we begin our study of the stability of the PWP method with respect to changes in the parameter $\lambda>0.$ Recall that whereas the matrix of direct influences $D$ comes directly from the user's knowledge and experience, the  parameter $\lambda$ comes from the PWP method itself. So, it is important to have a good control of the dependence of the PWP method on the choice of $\lambda.$\\

We consider the linear directed graph $\mathrm{L}_n$ with $n$ vertices. Before considering the general case we deal with three simple but illustrative examples, namely, the graphs $\mathrm{L}_2, \ \mathrm{L}_3, \ \mbox{and} \ \mathrm{L}_6$  shown in Figure \ref{l}.\\

 \begin{figure}[htb]
\centering
  \begin{tabular}{@{}ccc@{}}
  \includegraphics[width=0.3\textwidth]{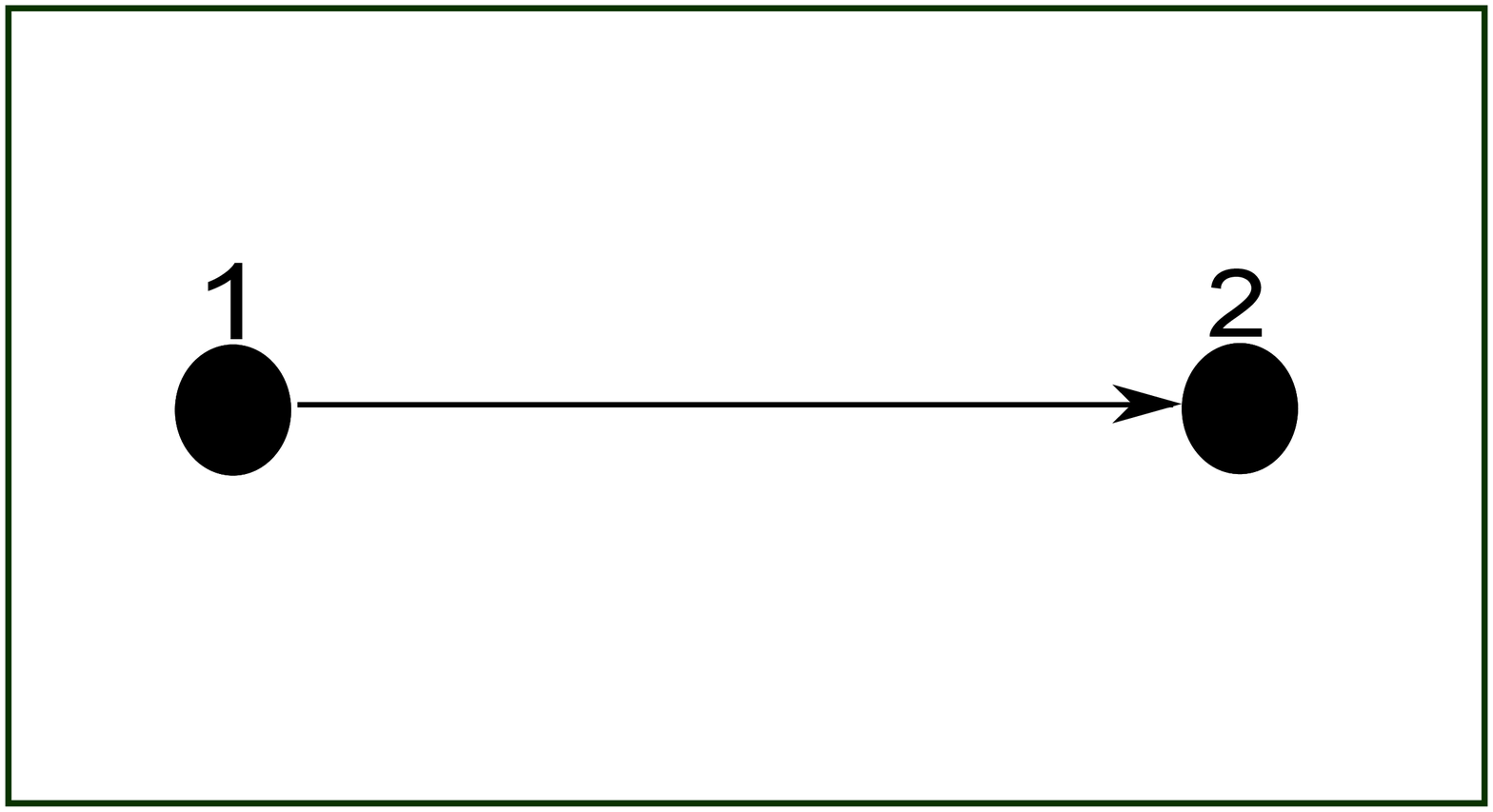}
 \includegraphics[width=0.3\textwidth]{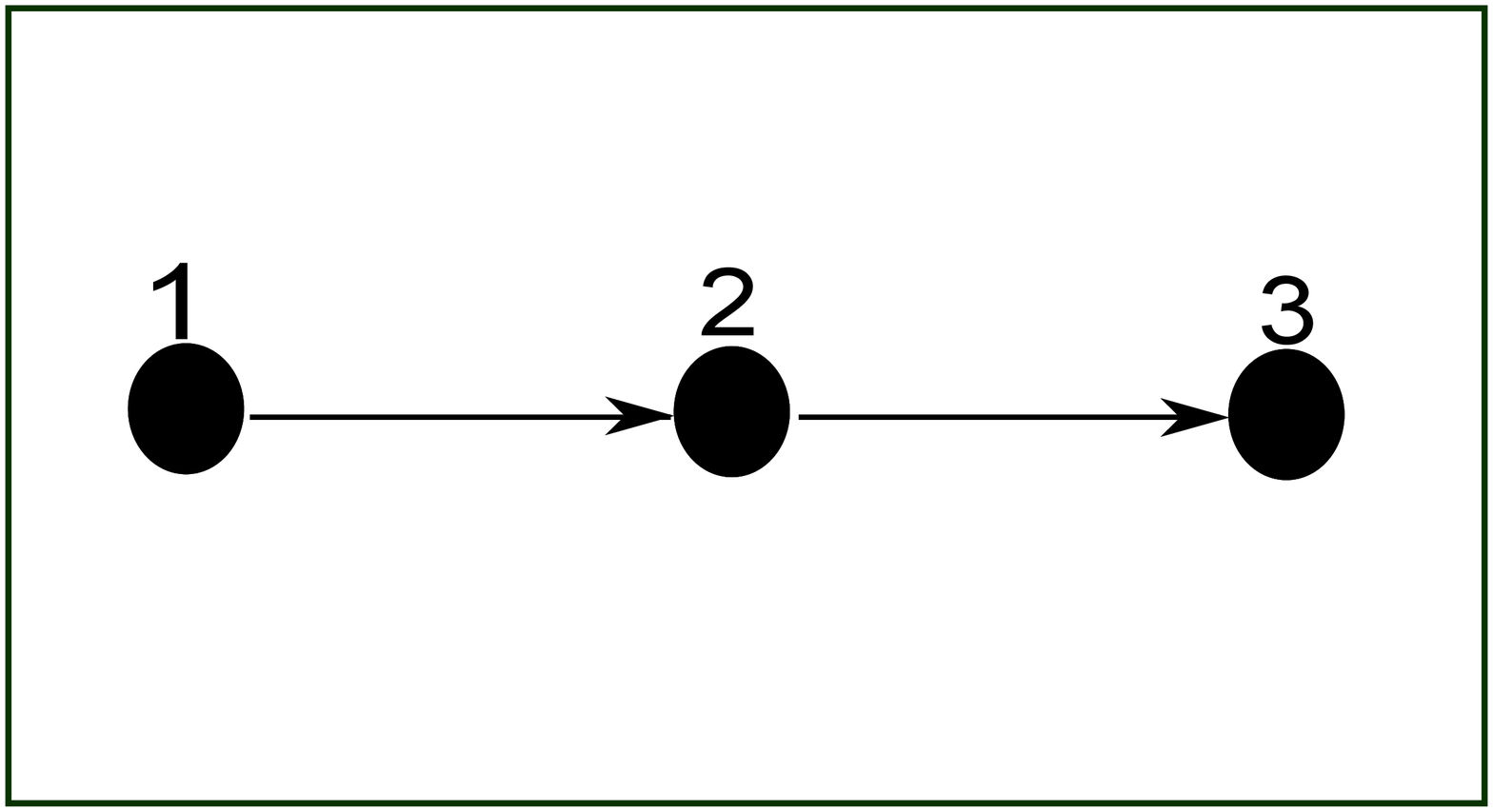}
 \includegraphics[width=0.3\textwidth]{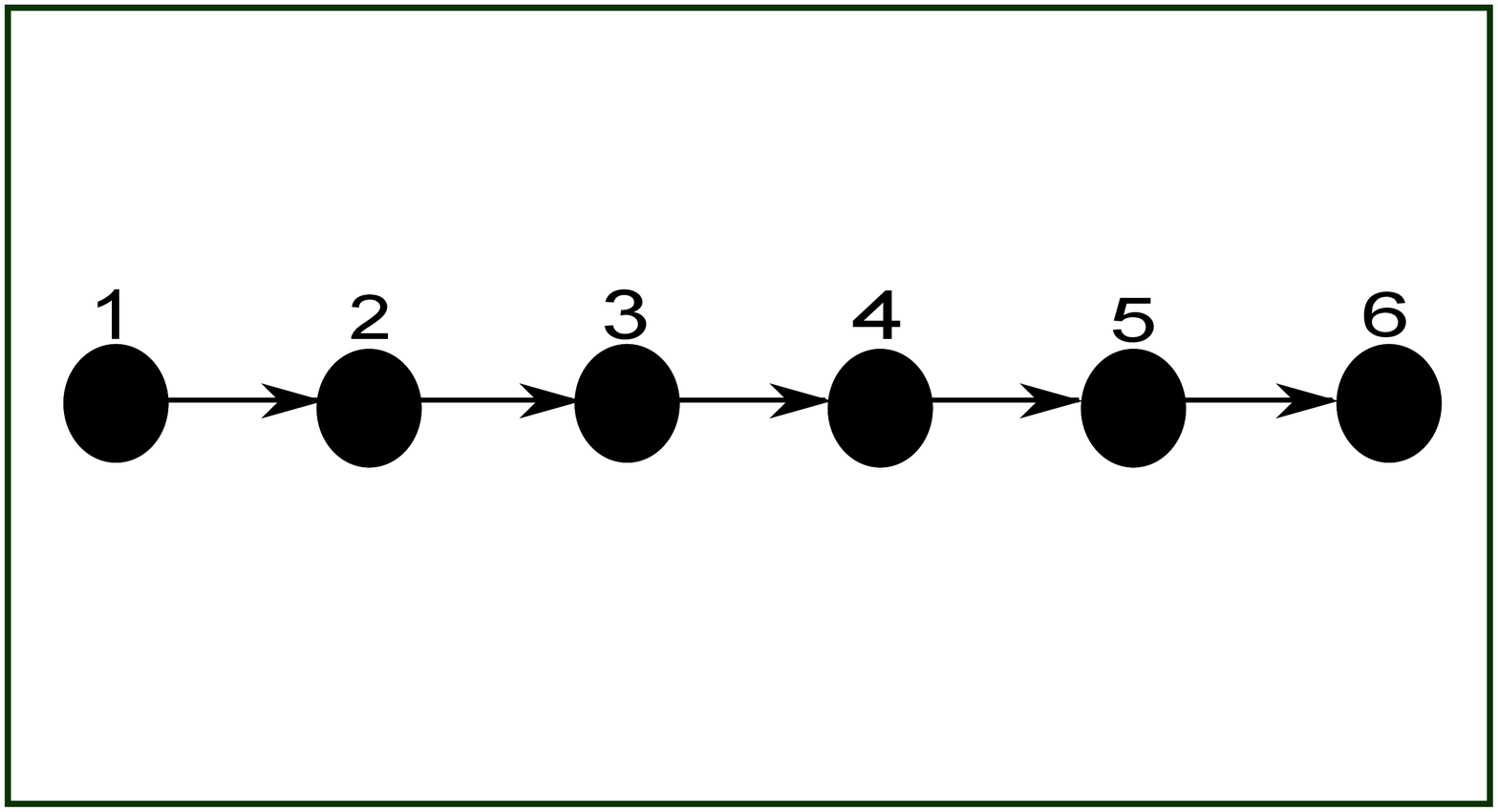}
   \end{tabular}
  \caption{Linear graphs $\mathrm{L}_2, \ \mathrm{L}_3 \ $ and $\mathrm{L}_6.$ }
  \label{l}
\end{figure}

 For $\mathrm{L}_2,$ it is easy to check that for all values of $\lambda$ both vertices $1$ and $2$ have the same importance given by $$\frac{\lambda}{e^{\lambda}-1} \ = \  \sum_{n=0}^{\infty}B_n \frac{\lambda^n}{n!},$$ where $B_n$ are the Bernoulli numbers \cite{diaz}. Thus as $\lambda$ goes to infinity the importance of both vertices goes to zero, meaning that the network becomes less connected, each vertex becomes more isolated, and thus its importance for the network decreases. See Figure \ref{curvas3}.\\

Regarding influences, vertex $1$ has influence $\frac{\lambda}{e^{\lambda}-1}$ while vertex $2$ has influence $0$. Thus vertex $1$ is always on top of vertex $2$ in the order of influences, however as $\lambda$ goes to infinity both influences tend to be equal to zero. Therefore for the graph $\mathrm{L}_2$ the PWP method is fully stable both in importance and in influence.\\

\begin{figure}[htb]
\centering
  \begin{tabular}{@{}ccc@{}}
   \includegraphics[width=0.3\textwidth]{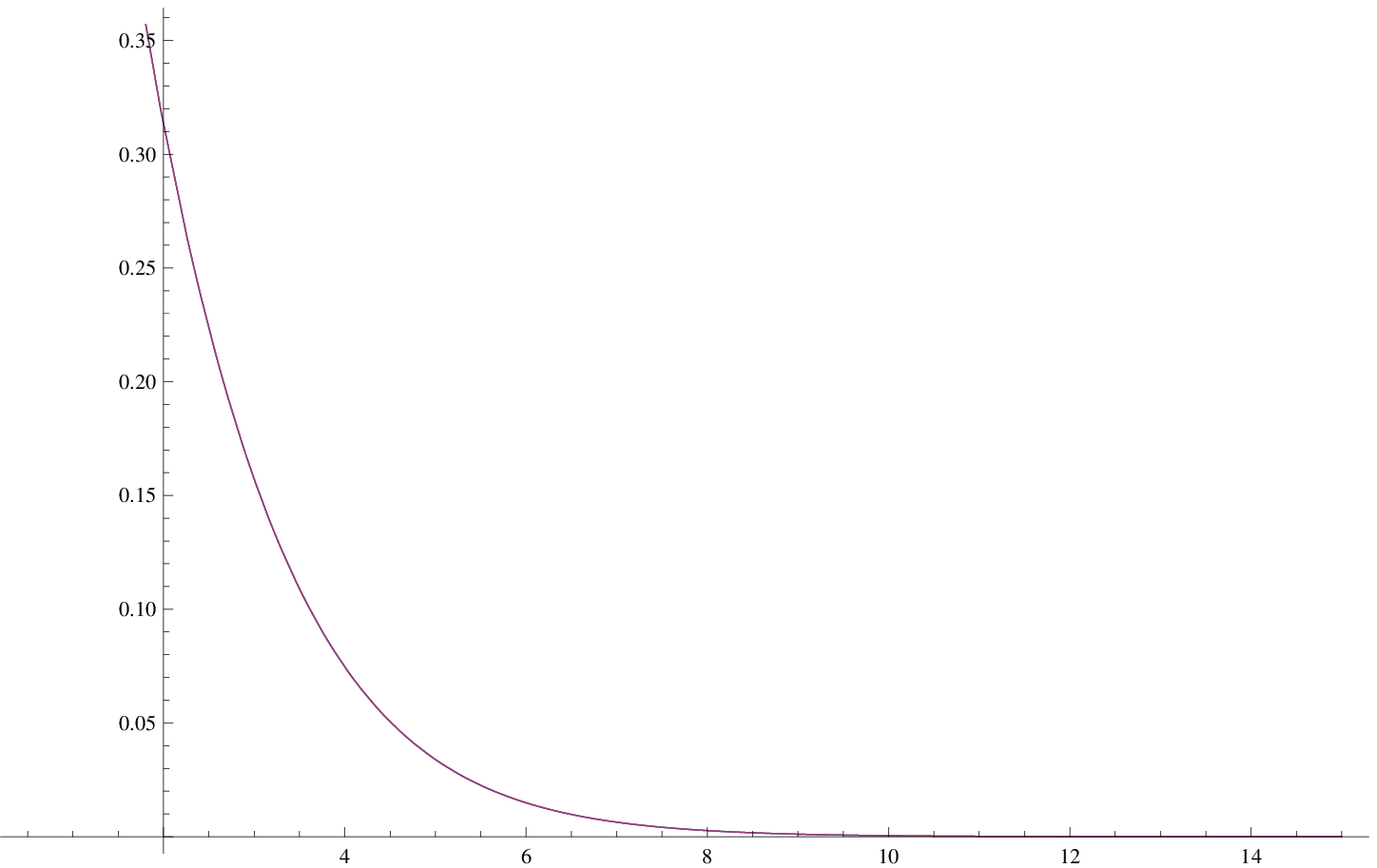}
   \includegraphics[width=0.3\textwidth]{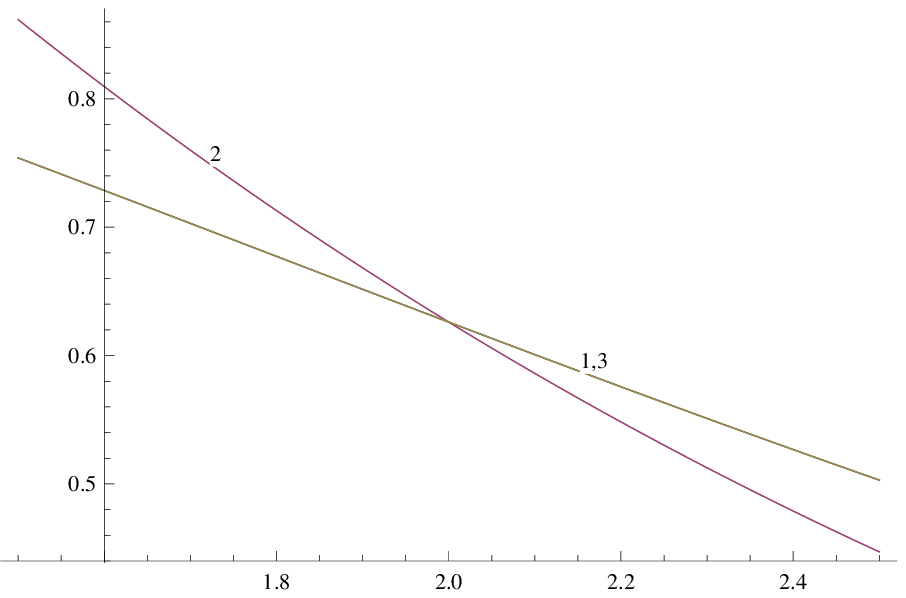}
  \includegraphics[width=0.3\textwidth]{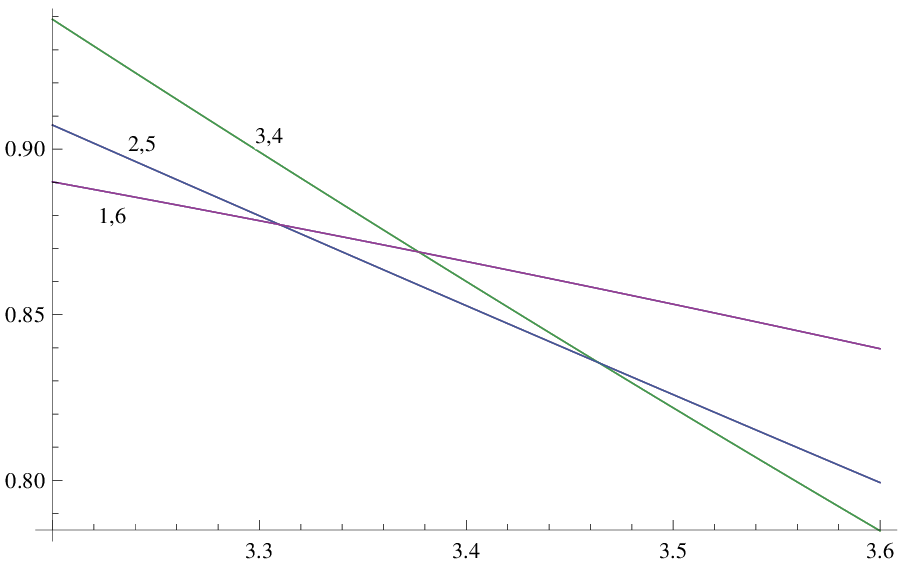}
   \end{tabular}
   \caption{Curves of importance for the graphs $\ \mathrm{L}_2, \ $ $\mathrm{L}_3,  $ and $ \mathrm{L}_6.$ }
   \label{curvas3}
\end{figure}

Let us now consider the case of the linear graph $\mathrm{L}_3$ with three vertices. Regarding importance, one can show that there is symmetry around the center of mass of the graph, as we will show  below happens for all linear graphs $\mathrm{L}_n$. Thus vertices $1$ and $3$ have the same importance, and therefore we only need to consider vertices $1$ and $2$. Here something interesting happens, as shown in Figure \ref{curvas3}, for small values of $\lambda$ vertex $2$ is the most important, but as $\lambda$ grows vertex $1$ overcomes vertex $2$ in importance. The matrix of indirect influences $T$ for $\mathrm{L}_3$ is given by:

$$T \ \ = \ \ \frac{1}{e_{+}^\lambda}\left(
\begin{array}{ccc}
 0 &  0 & 0 \\
 \lambda & 0 & 0  \\
 \frac{\lambda^{2}}{2} & \lambda & 0 \\
 \end{array}
 \right)$$
And thus the importance of vertex $1$ and $2$ are given, respectively, by:
$$I_1(\lambda) \ = \ \frac{2\lambda +  \lambda^2}{2e_{+}^\lambda} \ \ \ \ \ \ \ \mbox{and} \ \ \ \ \ \ \ I_2(\lambda)\ = \ \frac{2\lambda}{e_{+}^\lambda} .$$ Thus $\ I_1(2)=I_2(2), \ $ $\ I_1(\lambda) < I_2(\lambda) \  $ for $\ 0< \lambda < 2,  \ $ and $ \ I_2(\lambda) < I_1(\lambda) \  $ for $\ 2< \lambda.$ Therefore, for the  graph $\mathrm{L}_3$ the three possible rankings of its vertices by importance, after taking symmetry into account, actually occur, see Figure \ref{curvas3}. Influences in turn are fully stable as we have that
$$F_1(\lambda) \ = \ \frac{\lambda + \frac{ \lambda^2}{2}}{e_{+}^\lambda} \ \ > \ \  \frac{\lambda}{e_{+}^\lambda} \ = \ F_2(\lambda).$$

Let us consider the graph $\mathrm{L}_6$  with six vertices. The matrix of indirect influences is given by
$$e_{+}^\lambda T=\left(
\begin{array}{cccccc}
 0 & 0 & 0&0&0&0\\
 \lambda & 0 &0 &0&0&0\\
 \frac{\lambda^{2}}{2!} & \lambda & 0 &0&0&0\\
 \frac{\lambda^{3}}{3!} & \frac{\lambda^{2}}{2!} & \lambda &0&0&0\\
 \frac{\lambda^{4}}{4!} & \frac{\lambda^{3}}{3!} & \frac{\lambda^{2}}{2!}&
 \lambda &0&0\\
 \frac{\lambda^{5}}{5!} & \frac{\lambda^{4}}{4!} & \frac{\lambda^{3}}{3!}&
 \frac{\lambda^{2}}{2!} & \lambda &0\\
 \end{array}
 \right)$$
By symmetry it is enough to consider the vertices $1, 2$ and $3$, with importance given by
$$e_{+}^\lambda I_1(\lambda)  =  \lambda+\frac{\lambda^{2}}{2}+\frac{\lambda^{3}}{3!}+\frac{\lambda^{4}}{4!}+\frac{\lambda^{5}}{5!},  \ e_{+}^\lambda I_2(\lambda)  = 2\lambda+\frac{\lambda^{2}}{2}+\frac{\lambda^{3}}{3!}+\frac{\lambda^{4}}{4!},  \ e_{+}^\lambda I_3(\lambda) = 2\lambda+2\frac{\lambda^{2}}{2}+\frac{\lambda^{3}}{3!}. $$
Figure \ref{curvas3} shows the evolution of $I_3(\lambda)$ as $\lambda$ varies. For $\lambda$ small we have the ranking $\ 1 < 2 < 3 \ $ in importance. Thus, initially vertex $1$ is the less important one, and as $\lambda$ grows it first overcomes vertex $2$ and then overcomes vertex $3$ reaching the top position. Later on vertex $2$ overcomes vertex $3$, the ranking $\ 3 < 2 < 1\ $ is achieved, and it remains stable for large values of $\lambda.$ Indirect influences in turn are given by
$$e_{+}^\lambda F_1(\lambda)  =  \lambda+\frac{\lambda^{2}}{2}+\frac{\lambda^{3}}{3!}+\frac{\lambda^{4}}{4!}+\frac{\lambda^{5}}{5!}, \ \ e_{+}^\lambda F_2(\lambda)  = \lambda+\frac{\lambda^{2}}{2}+\frac{\lambda^{3}}{3!}+\frac{\lambda^{4}}{4!}, \ \ e_{+}^\lambda F_3(\lambda) = \lambda+\frac{\lambda^{2}}{2}+\frac{\lambda^{3}}{3!} ,$$ and thus the ranking $1 > 2 >3$ is stable for all values of $\lambda.$ Note however that the three values for importance approach $0$ as $\lambda$ goes to infinity.
Thus, although the comparative values change in order, the overall values converge to zero.\\

Next, we consider the  linear graph $\mathrm{L}_n$ with $n$ vertices. The matrix $D$ of direct influences and the matrix $T$ of indirect influences are given, respectively, by
$$D_{ij}\ = \
\left\{ \begin{array}{lcl}
1 \ \ \ \   \mbox{if}  \ \ i=j+1,\\
& & \\
0 \ \ \  \mbox{ otherwise},
\end{array}
\right.
\ \ \ \ \ \
e_{+}^\lambda T_{ij}\ = \
\left\{ \begin{array}{lcl}
\frac{\lambda^{i-j}}{(i-j)!} \ \ \ \   \mbox{if}  \ \ \ i > j,\\
& & \\
\ 0 \ \ \  \ \ \ \ \mbox{ otherwise}.
\end{array}
\right.$$

\begin{prop}\label{p1}{\em The importance of the vertex $j$ in the graph $\mathrm{L}_n$ is given by
$$e_{+}^\lambda I_j(\lambda) \ =  \ \sum_{i=1}^{j-1} \ \frac{\lambda^{j-i}}{(j-i)!}\  \ + \ \  \sum_{i=j+1}^{n} \ \frac{\lambda^{i-j}}{(i-j)!},$$ or equivalently:

\begin{enumerate}
  \item For $j=1,n, \ $ we have that $\ \ e_{+}^\lambda I_1(\lambda) \ = \ e_{+}^\lambda I_n(\lambda) \ = \ \lambda \ + \ \cdots \ + \ \frac{\lambda^{n-1}}{(n-1)!}.$

  \item For $\ 1 < j\leq \frac{n+1}{2}, \ $ we have that $$ e_{+}^\lambda I_{j}(\lambda) \  =  \ 2\lambda \ + \ \frac{2\lambda^{2}}{2}\ + \ \cdots \ + \ 2\frac{\lambda^{j-1}}{(j-1)!}\ + \ \frac{\lambda^{j}}{j!} \ + \ \cdots \ + \ \frac{\lambda^{n-j}}{(n-j)!},$$
  \item For $\ \frac{n+1}{2} \leq  j < n  , \ $ we have that
$$e_{+}^\lambda I_{j}(\lambda) \  =  \ 2\lambda\ + \ \cdots \ + \ 2\frac{\lambda^{n-j}}{(n-j)!}\ + \ \frac{\lambda^{n-j+1}}{(n-j+1)!}\ + \ \cdots \ + \ \frac{\lambda^{j-1}}{(j-1)!}.$$
\item The importance $I_j(\lambda) \rightarrow 0,$  as $\lambda \rightarrow \infty.$
\end{enumerate}
}
\end{prop}

\begin{proof} The first identity follows directly from the definitions.  The other three identities follow from the first after specialization and simple changes of variables. Part 4 follows from the previous formulae.
\end{proof}

Next result gives us the symmetry in importance around the center of mass  for the the graphs $\mathrm{L}_n$.

\begin{prop}\label{sym}{\em
The importance of vertices in the linear graph $\mathrm{L}_n$ is invariant under  the change $\ j  \longrightarrow n+1-j, \ $  i.e. we have for $j \in [n] $ that
$$I_{j}(\lambda) \ = \ I_{n+1-j}(\lambda).$$}
\end{prop}

\begin{proof} We already know from Proposition \ref{p1} that $I_1(\lambda) = I_n(\lambda).$
 Assume $\ 1 < j\leq \frac{n+1}{2},$ thus we have that $n > n+1-j \geq \frac{n+1}{2}$, and thus  Proposition \ref{p1} implies that:
$$e_{+}^\lambda I_{n+1-j}(\lambda) =  2\lambda+\cdots+\frac{2\lambda^{n-(n-j+1)}}{(n-(n-j+1))!}+
\frac{\lambda^{n-(n-j+1)+1}}{(n-(n-j+1)+1)!}+\cdots+\frac{\lambda^{n-j+1-1}}{(n-j)!} = $$
$$2\lambda\ + \ \cdots \ + \ \frac{2\lambda^{j-1}}{(j-1)!}\ + \ \frac{\lambda^{j}}{j!}\ + \ \cdots\ + \ \frac{\lambda^{n-j}}{(n-j)!} \ \ = \ \ e_{+}^\lambda I_{j}(\lambda).$$
\end{proof}

By symmetry we only need to consider vertices to the left of the center of mass: namely for $n=2k\ $ or $\ n=2k-1$ we only need to consider vertices $j$ such that $1 \leq j \leq k.$\\

Consider the maps of importance $\ I_j: (0,\infty) \longrightarrow \mathbb{R} \ $ and the corresponding curves
$$\{\ (\lambda, I_j(\lambda))\  \ | \ \ \lambda > 0 \ \}   \ \subseteq \ \mathbb{R}^2.$$ Our next result
shows that the order in importance for large $\lambda$ is the reverse of the order in importance for small $\lambda$.

\begin{lem}{\em For $n=2k \ $ or $\ n=2k-1$ we have that:

\begin{enumerate}
  \item If $\lambda$ is small enough, then $\ I_i(\lambda) < I_j(\lambda)\ $ for $\ 1 \leq i<j \leq k.$

  \item If $\lambda$ is large enough, then $\ I_i(\lambda) > I_j(\lambda)\ $ for $\ 1 \leq i<j \leq k.$
\end{enumerate}
}
\end{lem}

\begin{proof} According to Proposition \ref{p1}, for the smallest potency of $\lambda$ at which $I_i(\lambda)$ and  $I_j(\lambda) $ differ are, respectively, of the form $$\frac{\lambda^i}{i!} \ \ \ \ \ \mbox{and} \ \ \ \ \ 2\frac{\lambda^i}{i!} .$$
These terms control the behaviour of  $I_i(\lambda)$ and  $I_j(\lambda) $ for small $\lambda$, and thus $\ I_i(\lambda) < I_j(\lambda)$ since $i<j$. For large $\lambda$, we look for the largest powers in $\lambda$ in  $I_i(\lambda)$ and  $I_j(\lambda) $ which are given, respectively, by $$\frac{\lambda^{n-i}}{(n-i)!} \ \ \ \ \ \mbox{and} \ \ \ \ \ \frac{\lambda^{n-j}}{(n-j)!}.$$ Since $i<j$, we have that $\ I_i(\lambda) > I_j(\lambda)\ $ for large $\lambda.$
\end{proof}

Next we make a rather plausible statement which we have been able to verify numerically in many instances.

\begin{conj}\label{c}{\em Consider the linear graph $\mathrm{L}_n$, with $n=2k \ $ or $\ n=2k-1,$  and let $1 \leq i < j  \leq k.$ The curves of importance  $I_i(\lambda)$ and $I_j(\lambda)$ intersect each other in a unique point $\lambda=c_{ij} \in (0,\infty)$. These intersection points occur in the order
$$c_{il} \ < \ c_{jm} \ \ \ \ \mbox{for} \ \ \ \ \ i< l \leq k \ \ \ \mbox{and}\ \ \ j< m \leq k.$$}
\end{conj}

\

The meaning of Conjecture \ref{c} is that the curve $I_1(\lambda)$ begins at the bottom and crosses all other curves, first
$I_2(\lambda),$ then $I_3(\lambda)$ and so on, until it reaches the top; next the curve $I_2(\lambda)$ raises from the bottom to the second highest position, just below $I_1(\lambda),$ crossing the curves $I_3(\lambda),$ $I_4(\lambda),$ etc, in exactly that order. After all crossings have taken place the reverse order to the original one has been achieved, i.e. $I_1(\lambda) > I_2(\lambda) > ...$, and this order remains stable up to infinity.\\

The numerical evidence for Conjecture \ref{c} is quite solid. Figure \ref{1117} suggests that it holds for the linear graph
$\mathrm{L}_{11}\ $.\\

\begin{figure}
\centering
\includegraphics[width=0.4\textwidth]{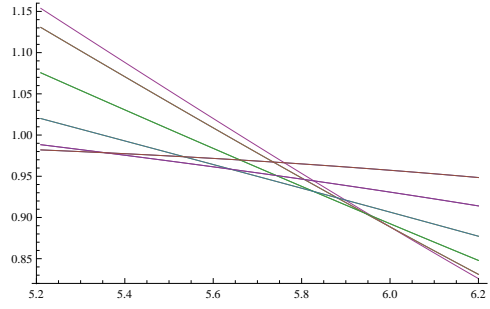}
\caption{Curves of Importance for $\ \mathrm{L}_{11}.\ $}
\label{1117}
\end{figure}

Next result allow us to locate the crossing points of consecutive curves of importance.

\begin{thm}{\em Let $\mathrm{L}_n$ be the linear graph with $n=2k \ $ or $\ n=2k-1, \ $ and $1 \leq i <  k-1.$ The curves of importance  $\ I_i(\lambda)\ $ and $\ I_{i+1}(\lambda)\ $ intersect each other in a unique point $\lambda = c_{i,i+1} \in (0,\infty)$ given by
$$c_{i,i+1}  \ = \  \Big( \frac{(n-i)!}{i!}\Big)^{\frac{1}{n-2i}}.$$
Thus for $\ n=2k\ $ and $\ n=2k-1\ $ we, respectively, have that $$ c_{i,i+1}  \ = \  \Big( \frac{(2k-i)!}{i!}\Big)^{\frac{1}{2(k-i)}} \ \ \ \ \ \ \ \mbox{and} \ \ \ \ \ \ \ c_{i,i+1}  \ = \  \Big( \frac{(2k-i-1)!}{i!}\Big)^{\frac{1}{2(k-i)-1}} .$$}
\end{thm}

\begin{proof} The crossing point $c_{i,i+1}$ is defined by the equation
$$2\lambda \ + \ \frac{2\lambda^{2}}{2}\ + \ \cdots \ + \ 2\frac{\lambda^{i-1}}{(i-1)!}\ + \ \frac{\lambda^{i}}{i!} \ + \ \cdots \ + \ \frac{\lambda^{n-i}}{(n-i)!} \ \ = $$ $$2\lambda \ + \ \frac{2\lambda^{2}}{2}\ + \ \cdots \ + \ 2\frac{\lambda^{i}}{i!}\ + \ \frac{\lambda^{i+1}}{(i+1)!} \ + \ \cdots \ + \ \frac{\lambda^{n-i-1}}{(n-i-1)!} ,$$ which after cancelling terms is equivalent to the equation
$$\frac{\lambda^{n-i}}{(n-i)!} \ \ = \ \ \frac{\lambda^i}{i!},$$ with a unique solution given by
$$c_{i,i+1} \ = \ \Big( \frac{(n-i)!}{i!}\Big)^{\frac{1}{n-2i}}.$$
\end{proof}

\begin{cor}{\em Let $\mathrm{L}_n$ be the linear graph with $n=2k \ $ or $\ n=2k-1$, and let $1 \leq i <  k-1.$
\begin{enumerate}
\item For $n=2k, \ k \geq 2, \ $ the crossing points $\ c_{1,2}\ $ and  $\ c_{k-1,k}\ $ are given by
$$ c_{1,2}  \ = \ (2k-1)!^{\frac{1}{2k-2}} \ \ \ \ \ \ \mbox{and} \ \ \ \ \ \ c_{k-1,k} \ = \ ( k(k+1))^{\frac{1}{2}}.$$
\item For $n=2k-1, \  k \geq 2, \ $ the crossing points $\ c_{1,2}\ $ and  $\ c_{k-1,k}\ $ are given by
$$c_{1,2}\ = \ (2k-2)!^{\frac{1}{2k-3}} \ \ \ \ \ \mbox{and} \ \ \ \ \ c_{k-1,k} \ = \ k.$$
\end{enumerate}
}
\end{cor}

The following result may be regarded as further evidence in favor of Conjecture \ref{c}.

\begin{thm}{\em The intersection points $c_{i,i+1}$ of the curves $I_i(\lambda)$ and $I_{i+1}(\lambda)$ occur in the order
$$c_{1,2}\ < \ c_{2,3} \ < \ \cdots \cdots \ < \ c_{k-2,k-1} \ < \ c_{k-1,k} .$$}
\end{thm}
\begin{proof}
We use a simple fact for positive integers: if $c<a-1$, then $$ac \ < \ (a-1)(c+1).$$
Choose $i$ such that $2(i+1) \leq n, \ $ i.e. such that $n-2i-2 \geq 0\ $ or equivalently $$ i+1 \ < \ n-i-1.$$  We have the following chain of equivalent inequalities
$$c_{i,i+1}\ < \ c_{i+1,i+2},$$
$$\Big( \frac{(n-i)!}{i!}\Big)^{\frac{1}{n-2i}} \ \ < \ \ \Big( \frac{(n-i-1)!}{(i+1)!}\Big)^{\frac{1}{n-2i-2}},$$
$$(n-i)!^{n-2i-2}(i+1)!^{n-2i}  \ \ < \ \ (n-i-1)!^{n-2i}i!^{n-2i-2} ,$$
$$(n-i)^{n-2i-2}(i+1)^{n-2i-2}(i+1)!^2 \ \ < \ \ (n-i-1)!^{2},$$
$$(n-i)^{n-2i-2}(i+1)^{n-2i-2} \ \ < \ \ (n-i-1)^{2} \cdots (i+2)^2,$$
$$((n-i)(i+1))^{n-2i-2} \ \ < \ \ (n-i-1)^{2} \cdots (i+2)^2$$
Using the fact mentioned at the beggining , it is clear that  in order to show the latter inequality it is enough to check that
$$  (n-i)(i+1) \ <  \ (n-i-1)(i+2),$$ which holds since $i+1 < n-i-1.$
\end{proof}

\begin{exmp}{\em Set $n=12\ $ and $\ i=3.$ In this case the inequality $$((n-i)(i+1))^{n-2i-2} \  <  \ (n-i-1)^{2} \cdots (i+2)^2$$ simply says that
$\ (9.4)^{4} \ < \ (8.7.6.5)^2, \  \mbox{or equivalently}$
$$(9.4)(9.4)(9.4)(9.4) \ < \ (8.5)(7.6)(7.6)(8.5).$$}
\end{exmp}

Note that if Conjecture \ref{c} holds, then there are exactly $$\frac{k(k-1)}{2}$$ crossing points among the curves of importance $I_j(\lambda)$, and thus that same number of different orderings by importance on the vertices of $\mathrm{L}_n$ as $\lambda$ varies, where $n=2k \ $ or $\ n=2k-1$. Indeed vertex $1$ begins as the less important and have to surpass $k-1$ vertices to reach the top. Then vertex $2$ have to surpass $k-2$ vertices to reach the second position, etc. Thus the number of crossing is the sum of the first $k-1$ natural numbers, yielding the desired result.\\

Although the number of reachable orderings grows to infinity, it is nevertheless a negligible quantity, for large $n$, relative to the number of all possible orderings as $$\underset{k \rightarrow \infty}{\mathrm{Lim}}\ \frac{k(k-1)}{2k!}\ = \ 0. $$ Therefore even though it is possible to reach many different orderings by choosing an appropriated $\lambda$, a random ordering will not be reachable.\\

According to Conjecture \ref{c}, for the linear graph $\mathrm{L}_n$ with $n=2k$ or $n=2k-1$, the first change in the ordering of vertices by importance occurs at $c_{1,2},$ while the last change occurs at $c_{k-1,k}$. The interval of stability $(0, c_{1,2})$ is dominated by the direct influences, while the interval
of stability $(c_{k-1,k}, \infty)$ is dominated by long indirect influences, yielding the reverse ordering. All reordering happens in the interval $[c_{1,2}, c_{k-1,k}].$\\

Finally, we consider the ranking of the vertices of  $\mathrm{L}_n$ by influence, and in this case we find full stability.

\begin{prop}{\em For any $\lambda >0,$ the ordering by indirect influences on the vertices of the linear graph $\mathrm{L}_n $  is given by $\ 1 > 2 > \cdots > n$.}
\end{prop}

\begin{proof}It follows since $\ F_n=0\ $ and for $\ 1 \leq i < n \ $ we have that:
$$F_i(\lambda) \ = \  \lambda \ + \ \frac{\lambda^{2}}{2}\ + \ \cdots \ + \ \frac{\lambda^{i-1}}{(i-1)!}\ + \ \frac{\lambda^{i}}{i!} \ + \ \cdots \ + \ \frac{\lambda^{n-i}}{(n-i)!}.$$
\end{proof}

Note however that while the ordering is completely stable, all influences approach $0$ as $\lambda$ goes to infinity.

\section{Stability on Circuits}\label{sw}

In the previous section we saw that linear chains of direct influences tend to have a destabilizing effect on the applications of the PWP method with respect to changes in $\lambda,$ to the point of allowing complete reversal in ordering for small and large values of $\lambda.$  In this section we are going to see that, in contrast, the presence of circuits have a stabilizing effect with respect to changes in $\lambda$.\\

Consider the circuit on the $\mathbb{Z}_n$ group of integers module $n$ given by
$$0 \ \longrightarrow \ 1 \ \longrightarrow \ \cdots \ \longrightarrow \ n-1 \ \longrightarrow 0.$$
For $i,j,\in \mathbb{Z}_n$ and $k \in [n]$, the matrices of direct and indirect influences are given by

$$D_{ij}\ = \
\left\{ \begin{array}{lcl}
1 \ \ \ \   \mbox{if}  \ \ i=j+1,\\
& & \\
0 \ \ \  \mbox{ otherwise},
\end{array}
\right.
\ \ \ \mbox{and} \ \ \ \ \ \ \ \
e_{+}^\lambda T_{j+k,j}\ = \ \sum_{l=0}^{\infty}\frac{\lambda^{k+ln}}{(k+ln)!}.
$$
Note that $T_{j+k,j}(\lambda)$ does not depend on $j$, thus we can use the simpler notation $T_k(\lambda).$
Next result shows the full stability of the PWP for circuits: stability in the rankings by importance, in influence, and even in the relative strength of indirect influences for small values of $\lambda$.

\begin{thm}{\em For $i,j \in \mathbb{Z}_n$ we have that:

\begin{enumerate}
\item All vertices in $\mathbb{Z}_n$ have equal importance, that is $\ I_i(\lambda)=I_j(\lambda).$
\item All vertices in $\mathbb{Z}_n$ have equal influences, that is $\ F_i(\lambda)=F_j(\lambda).$
\item For $\lambda \in (0,2)$ indirect influences are ordered as follows:
$$T_1(\lambda) \ > \ \cdots \ > \  T_n(\lambda) .$$
\end{enumerate}
}
\end{thm}

\begin{proof} Properties 1 and 2 follow from the identities
$$F_j(\lambda) \ = \  T_1(\lambda) \ + \ \cdots \ + \  T_n(\lambda)  \ \ \ \ \ \ \mbox{and} \ \ \ \ \ \ I_j(\lambda)=2F_j(\lambda).$$ Property 3 is shown as follows. Recall that for $k\in [n-1]$ we have $$ T_{k}(\lambda)\ = \ \sum_{l=0}^{\infty}\frac{\lambda^{k+ln}}{(k+ln)!},$$ thus the inequality
$ T_{k}(\lambda) > T_{k+1}(\lambda)$ holds for each summand in the respective series expansions if and only if
$$ \frac{\lambda^{k+ln}}{(k+ln)!} \ > \ \frac{\lambda^{k+1 +ln}}{(k+1 +ln)!},$$ or equivalently
$$\lambda \ < \  k+1+ln.$$ Thus for $\lambda <2,$ we have that
$$\lambda \ < \ 2 \ = \ 1+1 \ \leq \ k+1 \ \leq \ k+1+ln, $$ and we obtain the desired inequality.
\end{proof}

Thus we see that the presence of circuits in a complex network have a stabilizing effect in the applications of the PWP method with respect to changes in  $\lambda.$

\section{Stability on $\mathbb{R}$-Diagonalizable Networks}

Assume that our matrix of direct influences $D$ is diagonalizable in $\mathbb{R}$, i.e. there is an invertible matrix $A \in \mathrm{M}_n(\mathbb{R})$ such that
$D=AEA^{-1}$  where $E$ is a diagonal matrix whose entries $E_{ii}$ give the eigenvalues of $D$. This condition holds, for example, if the eigenvalues of $A$ are all real and distinct, a generic condition among matrices with only real eigenvalues. \\

The PWP map is equivariant with respect to conjugation, thus we have that
$$\mathrm{T}(D, \lambda) \ = \ A\mathrm{T}(E, \lambda)A^{-1} .$$ The matrix $\mathrm{T}(E, \lambda)$ is rather simple to compute, indeed it is a diagonal matrix with entries
$$\mathrm{T}(E, \lambda)_{ii} \ = \ \frac{e_+^{E_{ii}}}{e_+^{\lambda}}. $$

\begin{thm}{\em For a  $\mathbb{R}$-diagonalizable  network there can be only a finite number of changes in the ranking of the vertices of the network by importance or by influence.}
\end{thm}

\begin{proof}
We must show that only a finite number of crossing points may occur, both for the curves of importance and the curves of indirect influences.
Since $D$ is a diagonalizable matrix in $\mathbb{R}$ it has $n$ real eigenvalues (counted with multiplicity). Assume it has $m$ different eigenvalues which we write in increasing order $d_1 < \ldots < d_m$. The key observation is that the entries of the matrix $$e_+^{\lambda}\mathrm{T}(D, \lambda) \ = \ e_+^{\lambda}A\mathrm{T}(E, \lambda)A^{-1}$$ are all of the form $$a_1 e_+^{d_1\lambda} \ + \ \cdots \ + \ a_m e_+^{d_m\lambda},$$ and therefore the functions of importance $\ e_+^{\lambda} I_j(\lambda),\ $ and the functions of influence $\ e_+^{\lambda} F_j(\lambda)\ $ are also of the same form. Therefore, finding the crossing points for the curves of importance or influence boils down to finding solutions to equations of the form $$a_1 e_+^{d_1\lambda} \ + \ \cdots \ + \ a_m e_+^{d_m\lambda} \  = \ 0.$$ Some of the resulting equations may be trivial, i.e. all the coefficients $a_i$ may be zero, meaning that some of the importance or influences functions are identically equal, which actually reduces the scope of possibilities for crossing points. We are going to show that as soon as one of these equations is non-trivial there can only be a finite number of solutions.
Since $e_+^x = e^x -1,$ the equation above is equivalent to a equation of the form
\begin{equation}
  \tag{E}
  a_1 e^{d_1\lambda} \ + \ \cdots \ + \ a_m e^{d_m\lambda} \ = \ a
  \label{eqn:eq}
\end{equation}
with $a=a_1 \ + \ \ldots \ + a_m. $
Without lost of generality we assume that $a_m \neq 0.$ If it is the only non-vanishing coefficient, then  \eqref{eqn:eq} reduces to $a_me^{d_m\lambda}=a_m$ which is either trivial if $d_m=0,$ or has no solution at all.\\

If another coefficient besides $a_m$ is non-zero, we may assume without lost of generality that $a_1 \neq 0.$ We show that if \eqref{eqn:eq} has infinitely many solutions they must be contained in a bounded interval around $0$. Assume that $d_m > 0,$ otherwise all the eigenvalues $d_i$ must be negative, and then as $t \rightarrow 0$ we have that  $$a_1 e^{d_1\lambda} \ + \ \cdots \ + \ a_m e^{d_m\lambda}\ \rightarrow \ 0,$$ and therefore \eqref{eqn:eq} has no solutions for large $\lambda $ if $ a \neq 0.$ If $a=0$, then \eqref{eqn:eq} is equivalent to
$$a_2 e^{(d_2-d_1)\lambda} \ + \ \cdots \ + \ a_m e^{(d_m-d_1)\lambda} \ = \ -a_1,$$
and we are back in the case where the last coefficient $d_m-d_1$ is positive. \\

Dividing  \eqref{eqn:eq} by $e^{d_m\lambda}$ and letting $\lambda$ go to infinity we find that $a_m=0,$ a contradiction, and thus no large $\lambda$ can be a solution of \eqref{eqn:eq}. So, if \eqref{eqn:eq} have infinitely many solutions they must be contained in an interval around $0$, and these solutions must have an accumulation point. Let us show that this point of accumulation can not be $0.$ We argue by contradiction. Assume that  \eqref{eqn:eq} has infinitely many solutions $t_l$ with $t_l \rightarrow 0\ $ as $\ l \rightarrow \infty.$ Then the derivative of the left-hand side of \eqref{eqn:eq} will have infinitely many zeroes accumulating at $0$, and thus the second derivative will also have infinitely many zeroes accumulating at $0$, and so on ... Therefore the coefficients $a_1, \ldots, a_m$ must be such that the identities
$$a_1d_1^k \ + \ \cdots \ + \ a_m d_m^k=0 ,$$ hold for $k \in \mathbb{N}_{>0}.$ Dividing the equation above by $d_m^k$ and letting $k$ go to infinity we find that $a_m=0,$ a contradiction. \\

Finally, assume that our infinitely many solutions $t_l$ of \eqref{eqn:eq}  have an accumulation point $c \neq 0.$ Set $t_l = c + s_l$, then the points $s_l$ give infinitely many solutions to the equation
$$(a_1 e^{d_1 c})e^{d_1 \lambda} \ + \ \cdots \ + \ (a_m e^{d_m c})e^{d_m \lambda} \ = \ a$$
accumulating at $0$. Therefore we must have that $$a_1 e^{d_1 c}\ = \ \cdots \ = \ a_m e^{d_m c} \ = \ 0,$$ and so
$\ a_1 =  \cdots  = a_m =  0.$
 \end{proof}

\begin{rem}{\em If we allow complex eigenvalues, then besides exponential functions, trigonometric functions may appear in the calculation of the matrix of indirect influences. Trigonometric functions can have infinitely many crossing points.}
\end{rem}

Thus for a diagonalizable network only a finite number of changes in ordering either by importance or influence can occur, and thus for such networks there
exists $\lambda_d \in (0, \infty)$ where the first change in order occurs, and a $\lambda_i \in (0, \infty)$ where the last change in order occurs. Thus the PWP method for $\lambda$  in the interval $(0, \lambda_d)$ is dominated by the direct influences, and the PWP method for $\lambda$ in the interval $(\lambda_i,\infty)$ is dominated by the indirect influences. All reordering happens in the interval $[\lambda_d, \lambda_i].$

\section{Conclusion}

In this work we have considered the stability of the PWP method for ranking vertices in a complex network by importance and indirect influences.
We have found that the PWP method is quite sensitive with respect to data, a fact that we regard as being positive. Stability with respect to the parameter $\lambda$ seems to involve, at least, two opposite forces. On the one hand, long directed path increases instability, to the extreme of allowing full order reversal, and other hand the presence of circuits tends to stabilize the applications of the method, again to the extreme of allowing full uniformization both in importance and influence. Must networks, of course, include both directed paths and circuits, and the stability of the PWP method with respect to changes in $\lambda$ will involve a subtle balance between these opposite forces. We have shown that for a  network diagonalizable in $\mathbb{R}$ only a finite number of changes in ranking may occur, both in importance and in influence. In contrast, the presence of complex eigenvalues opens the door for infinitely many changes in order. The next challenge, left for future research, is to study the stability of the PWP method for randomly generated networks.

\section*{Acknowledgements}
We thank Jorge Catumba for many helpful comments and suggestions. Numerical computations in this work were done using the Scilab module "Indirect Influences for Graphs" developed by Jorge Catumba \cite{pag2}.

\

\

\noindent ragadiaz@gmail.com, \\
\noindent Departamento de Matem\'aticas,\\
\noindent Universidad Javeriana, \ Bogot\'a, \ Colombia\\

\noindent angelikius90@gmail.com  \\
\noindent Escuela de Ciencias Exactas e Ingener\'ias,\\
\noindent Universidad Sergio Arboleda, \ Bogot\'a, \ Colombia\\


\begin{thebibliography}{10}

\bibitem{b2}
S. Brin,  L. Page, R. Motwani,  T. Winograd, The Anatomy of a Large-Scale Hypertextual Web Search Engine, Comp. Netw. ISDN Sys. 30 (1998) 107-117.

\bibitem{b1}
S. Brin, L. Page, R. Motwani,  T. Winograd, The PageRank citation ranking: Bringing order to the Web,
Technical Report, Stanford Digital Library Technologies Project (1998).


\bibitem{pag2}
 J. Catumba, Indirect Influences for Graphs, http:atoms.scilab.org/toolboxes/indinf.

\bibitem{chung}
F. Chung, The heat kernel as the pagerank of a graph, Proc. Natl. Acad. Sci. U.S.A.  104 (2007) 19735-19740.

\bibitem{diaz}
R. D\'iaz, Indirect Influences, Adv. Stud. Contemp. Math. 23 (2013) 29-41.

\bibitem{diazgomez}
R. D\'iaz, L. G\'omez,  Indirect Influences in International Trade, Netw. Heterog. Media 10 (2015) 149-165.

\bibitem{estrada}
E. Estrada, N. Hatano, Communicability in complex networks, Phys. Rev. E 77 (2008) 036111.

\bibitem{godet}
M. Godet,  De l'Anticipation $\grave{\mbox{a}}$ l'Action, Dunod, Par\'is 1992.

\bibitem{k} L. Katz, A new status index derived from sociometric analysis, Psychmetrika 18 (1953) 39-43.

\bibitem{meyer}
A. Langville, C. Meyer, Deeper Inside PageRank, Internet Math. 1 (2004) 335-400.

\end{thebibliography}
\end{document}